\documentclass[copyright,creativecommons]{eptcs}
\providecommand{\event}{FICS 2013} 

\usepackage{amsmath,amsfonts,amssymb,latexsym,amsthm}
\usepackage{xcolor} 
\usepackage{stmaryrd}
\usepackage{hyperref}

\newtheorem{theorem}{Theorem}[section]
\newtheorem{lemma}[theorem]{Lemma}
\theoremstyle{remark}
\newtheorem{remark}{Remark}

\newcommand\fun{\mathrel\rightarrow}
\newcommand\subst[2]{[#1/#2]}
\newcommand\rebind[2]{[#1/#2]}
\newcommand\Imp{\mathrel\Rightarrow}
\newcommand\Iff{\mathrel\Leftrightarrow}
\newcommand\eps\varepsilon
\newcommand\subs\subseteq
\newcommand\set[2]{\left\{#1\mid #2\right\}}
\newcommand\mathname[1]{\ensuremath{\mathsf{#1}}}
\newcommand\Set{\mathname{Set}}
\newcommand\sem[1]{\llbracket{#1}\rrbracket}
\newcommand\cset[1]{\{#1\}}
\newcommand\FV[1]{\mathname{FV}(#1)}
\newcommand\CFname{\mathname{CF}}
\newcommand\CF[1]{\CFname\,#1}
\newcommand\Tname{\mathname{T}}
\newcommand\T[1]{\Tname\,#1}
\newcommand\KAF{\mathname{KAF}}
\newcommand\KAR{\mathname{KAR}}
\newcommand\KAG{\mathname{KAG}}

\newcommand{\fhcomment}[1]{\textcolor{red}{[\textbf{Comment (FH)}: {#1}]}}
\newcommand{\dkcomment}[1]{\textcolor{blue}{[\textbf{Comment (DK)}: {#1}]}}
\renewcommand{\fhcomment}[1]{}
\renewcommand{\dkcomment}[1]{}

\begin{document}

\title{Infinitary Axiomatization of the Equational Theory of Context-Free Languages}
\author{Niels Bj\o rn Bugge Grathwohl
\institute{Department of Computer Science (DIKU)\\
University of Copenhagen\\
Universitetsparken 5\\
DK-2100 Copenhagen, Denmark}
\email{bugge@diku.dk}
\and
Fritz Henglein
\institute{Department of Computer Science (DIKU)\\
University of Copenhagen\\
Universitetsparken 5\\
DK-2100 Copenhagen, Denmark}
\email{henglein@diku.dk}
\and
Dexter Kozen
\institute{Department of Computer Science\\
Cornell University\\
Ithaca, NY 14853-7501, USA}
\email{kozen@cs.cornell.edu}
}
\def\titlerunning{Infinitary Axiomatization of the Equational Theory of Context-Free Languages}
\def\authorrunning{N.B.B.~Grathwohl, F.~Henglein, D.~Kozen}

\maketitle

\begin{abstract}
We give a natural complete infinitary axiomatization of the equational theory of the context-free languages, answering a question of Lei\ss\ (1992).
\end{abstract}

\section{Introduction}
\label{sec:intro}

Algebraic reasoning about programming language constructs has been a popular research topic for many years. At the propositional level, the theory of flowchart programs and linear recursion are well handled by such systems as 
\fhcomment{Write: affine instead of linear}%
\dkcomment{I believe linear recursion is standard terminology.\\
http://www.csse.monash.edu.au/~lloyd/tildeAlgDS/Recn/Linear/ \\
http://mitpress.mit.edu/sicp/chapter1/node12.html \\
http://www.sparknotes.com/cs/recursion/whatisrecursion/section2.rhtml}%
\fhcomment{You win.  How did we ever end up calling something like $a x + b$, with $b \neq 0$, a linear function...}%
Kleene algebra and iteration theories, systems that characterize the equational theory of the regular sets. To handle more general forms of recursion including procedures with recursive calls, one must extend to the context-free languages, and here the situation is less well understood. One reason for this is that, unlike the equational theory of the regular sets, the equational theory of the context-free languages is not recursively enumerable. This has led some researchers to declare its complete axiomatization an insurmountable task \cite{leiss92b}.

Whereas linear recursion can be characterized with the star operator $^\star$ of Kleene algebra or the dagger operation $^\dagger$ of iteration theories, the theory of context-free languages requires a more general fixpoint operator $\mu$. The characterization of the context-free languages as least solutions of algebraic inequalities involving $\mu$ goes back to a 1971 paper of Gruska \cite{Gruska71}. More recently, several researchers have given equational axioms for semirings with $\mu$ and have developed fragments of the equational theory of context-free languages \cite{cou86,esle2002,EsikLeiss05,Hopkins08a,Hopkins08b,leiss92b}.

In this paper we consider another class of models satisfying a condition called \emph{$\mu$-continuity} analogous to the star-continuity condition of Kleene algebra:
\begin{align*}
a(\mu x.p)b &= \sum_{n\geq 0}\,a(nx.p)b,
\end{align*}
where the summation symbol denotes supremum with respect to the natural order in the semiring, and
\begin{align*}
0x.p &= 0 & (n{+}1)x.p &= p\subst x{nx.p}.
\end{align*}
This infinitary axiom combines the assertions that $\mu x.p$ is the supremum of its finite approximants $nx.p$ and that multiplication in the semiring is continuous with respect to these suprema. Analogous to a similar result for star-continuous Kleene algebra, we show that all context-free languages over a $\mu$-continuous idempotent semiring have suprema. Our main result is that the $\mu$-continuity condition, along with the axioms of idempotent semirings, completely axiomatize the equational theory of the context-free languages. This is the first completeness result for the equational theory of the context-free languages, answering a question of Lei\ss\ \cite{leiss92b}.

\subsection{Related Work}
 
Courcelle \cite{cou86} investigates \emph{regular systems}, finite systems of fixpoint equations over first-order terms over a ranked alphabet with a designated symbol $+$ denoting set union, thereby restricting algebras to power set algebras.  He stages their interpretation by first interpreting recursion over first-order terms as infinite trees, essentially as the final object in the corresponding coalgebra, then interpreting the signature symbols in $\omega$-complete algebras.  He provides soundness and completeness for transforming regular systems that preserve all solutions and soundness, but not completeness for preserving their least solutions.  Courcelle's approach is syntactic since it employs unfolding of terms in fixpoint equations.

Lei\ss\ \cite{leiss92b} investigates three classes of idempotent semirings with a syntactic least fixpoint operator $\mu$. The three classes are called \KAF, \KAR, and \KAG\ in increasing order of specificity. All these classes are assumed to satisfy the fundamental \emph{Park axioms}
\begin{align*}
p\subst x{\mu x.p} &\leq \mu x.p & p \leq x\ &\Imp\ \mu x.p \leq x, 
\end{align*}
which say that $\mu x.p$ is the least solution of the inequality $p \leq x$. The classes \KAR\ and \KAG\ further assume
\begin{align*}
\mu x.(b + ax) &= \mu x.(1 + xa)\cdot b & \mu x.(b + xa) &= b\cdot\mu x.(1 + ax)
\end{align*}
and
\begin{align*}
\mu x.(s + rx) &= \mu x.(\mu y.(1 + yr)\cdot s) & \mu x.(s + xr) &= \mu x.(s\cdot \mu y.(1 + ry)),
\end{align*}
respectively. These axioms can be viewed as imposing continuity properties of the semiring operators with respect to $\mu$. All standard interpretations, including the context-free languages over an alphabet $X$, are continuous and satisfy the \KAG\ axioms. \'Esik and Lei\ss\ \cite{esle2002,EsikLeiss05} show that conversion to Greibach normal form can be performed purely algebraically under these assumptions.

\'Esik and Kuich \cite{esik2007modern} introduce \emph{continuous semirings}, which are required to have suprema for all directed sets, and they
employ domain theory to solve polynomial fixpoint equations.  Idempotent  continuous semirings are $\mu$-continuous Chomsky algebras as 
defined here, but not conversely.  As we shall prove, the family of context-free languages over any alphabet constitutes a $\mu$-continuous Chomsky algebra.  It is not a continuous semiring, however, since the union of context-free languages is not necessarily context-free.   

\section{Chomsky Algebras}
\label{sec:CF}

\subsection{Polynomials}

Let $(C,\,+,\,\cdot,\,0,\,1)$ be an idempotent semiring and $X$ a fixed set of variables. A \emph{polynomial over indeterminates $X$ with coefficients in $C$} is an element of $C[X]$, where $C[X]$ is the coproduct (direct sum) 
\fhcomment{Add: ``categorical'' before coproduct}%
\dkcomment{Why? What other kind is there?}%
\fhcomment{Good question.  I was thinking of categorical sum.}%
of $C$ and the free idempotent semiring on generators $X$ in the category of idempotent semirings. For example, if $a,b,c\in C$ and $x,y\in X$, then the following are polynomials:
\begin{align*}
&& 0 && a && axbycx + 1 && ax^2byx + by^2xc && 1 + x + x^2 + x^3
\end{align*}
The elements of $C[X]$ are not purely syntactic, as they satisfy all the equations of idempotent semirings and identities of $C$. For example, if 
$a^2=b^2=1$ in $C$, then
\begin{align*}
(axa + byb)^2 &= ax^2a + axabyb + bybaxa + by^2b.
\end{align*}
Every polynomial can be written as a finite sum of monomials of the form
\begin{align*}
a_0x_0a_1x_1\cdots a_{n-1}x_{n-1}a_n,
\end{align*}
where each $a_i\in C-\{0\}$ and $x_i\in X$. The \emph{free variables} of such an expression $p$ are the elements of $X$ appearing in it and are denoted $\FV p$. The representation is unique up to associativity of multiplication and associativity, commutativity, and idempotence of addition.

\fhcomment{Problem: ``Free variables'' not well-defined for $C[X]$.
Proposal: Change from $C[X]$ to introduction of polynomial expressions 
$\E X$, 
interpretation over polynomial expressions, analogous to $\T X$ of
Section 2.4.  We introduce ``free'' $\mu$-expressions in this version of the paper and the domain of $\sigma$ need not be a semiring anyway; gives simpler presentation since $C[X]$ not required, only expressions $\E X$ and their extension $\T X$ in Section 2.4.}%
\dkcomment{Actually, I believe it is well defined if you write the polynomial as a sum of monomials as I have described.}%
\fhcomment{I see, thanks. Should we add that every polynomial can be written \emph{uniquely} as a sum over a set of monomials?}%

\subsection{Polynomial Functions and Evaluation}

Let $C[X]$ be the semiring of polynomials over indeterminates $X$ and let $D$ be an idempotent semiring containing $C$ as a subalgebra.
By general considerations of universal algebra, any 
valuation $\sigma:X\fun D$ extends uniquely to a semiring homomorphism $\hat\sigma:C[X]\fun D$ preserving $C$ pointwise.
Formally, the functor $X\mapsto C[X]$ is left adjoint to a forgetful functor that takes an idempotent semiring $D$ to its underlying set.
Intuitively, $\hat\sigma$ is the \emph{evaluation morphism} that evaluates a polynomial at the point $\sigma\in D^X$.
Thus each polynomial $p\in C[X]$ determines a \emph{polynomial function} $\sem p:D^X\fun D$, where $\sem p(\sigma)=\hat\sigma(p)$.

The set of all functions $D^X\fun D$ with the pointwise semiring operations is itself an idempotent semiring with $C$ as an embedded subalgebra under the embedding $c\mapsto\lambda\sigma.c$. The map $\sem\cdot:C[X]\fun (D^X\fun D)$ is actually $\hat\tau$, where  $\tau(x) = \lambda f.f(x)$.

For the remainder of the paper, we write $\sigma$ for $\hat\sigma$, as there is no longer any need to distinguish them.

\subsection{Algebraic Closure and Chomsky Algebras}
\label{sec:Chomsky}

A \emph{system of polynomial inequalities over $C$} is a set
\begin{align}
p_1\leq x_1,\ p_2\leq x_2,\ \ldots,\ p_n\leq x_n\label{eq:sysineq}
\end{align}
where $x_i\in X$ and $p_i\in C[X]$, $1\leq i\leq n$. A \emph{solution} of \eqref{eq:sysineq} in $C$ is a valuation $\sigma:X\fun C$ such that $\sigma(p_i)\leq \sigma(x_i)$, $1\leq i\leq n$. The solution $\sigma$ is a \emph{least solution} if $\sigma\leq\tau$ pointwise for any other solution $\tau$. If a least solution exists, then it is unique.

An idempotent semiring $C$ is said to be \emph{algebraically closed} if every finite system of polynomial inequalities over $C$ has a least solution in $C$.

The category of \emph{Chomsky algebras} consists of algebraically closed idempotent semirings along with semiring homomorphisms that preserve least solutions of systems of polynomial inequalities.
\fhcomment{Add: Definition of what exactly this means.}%
\dkcomment{I didn't do this because it is more or less obvious what we mean, but clunky (and not too instructive) to write down formally, and we don't really use it anywhere. It would have been moot if we had started off with $\mu$-expressions in our language, because homomorphisms would have to preserve $\mu$. But I like our formulation better because it is closer in spirit to the definition of algebraically closed fields. The corresponding statement in fields is that automorphisms preserve roots of polynomials. Formally, a Chomsky algebra morphism is a semiring homomorphism $h:C\fun D$ such that if $h':C[X]\fun D[X]$ is the unique extension of $h$ such that $h'(x)=x$, and if $\sigma$ is the least solution to the system $p_i\leq x_i$, $\onein$, then $h\circ\sigma$ is the least solution to the system $h'(p_i)\leq x_i$, $\onein$.}%
\fhcomment{Okay.  A reader might be interested anyway, given 
it's mentioned to start with.  Going with both $\E X$ and $\T X$ as expressions in this paper could 
save space, I believe, and leave an proper algebraic treatment of $\mu$ as
extension to $C[X]$ as something for later.  What do you think?}%

The canonical example of a Chomsky algebra is the family of context-free languages $\CF X$ over an alphabet $X$. A system of polynomial inequalities \eqref{eq:sysineq} can be regarded as context-free grammar, and the least solution of the system is the context-free language generated by the grammar. For example, the set of strings in $\{a,b\}^\star$ with equally many $a$'s and $b$'s is generated by the grammar
\begin{align}
S &\fun \eps \mid aB \mid bA & A &\fun aS \mid bAA & B &\fun bS \mid aBB,\label{ex:Greibach1}
\end{align}
which corresponds to the system
\begin{align}
1+aB+bA &\leq S & aS+bAA &\leq A & bS+aBB &\leq B,\label{ex:Greibach2}
\end{align}
where the symbols $a,b$ are interpreted as the singleton sets $\cset a,\cset b$, the symbols $S,A,B$ are variables ranging over sets of strings, and the semiring operations $+$, $\cdot$, $0$, and $1$ are interpreted as set union, set product $AB = \set{xy}{x\in A,\ y\in B}$, $\emptyset$, and $\cset\eps$, respectively.
 
\subsection{$\mu$-Expressions}
\label{sec:mu-exp}

Let $X$ be a set of indeterminates.
Lei\ss\ \cite{leiss92b} and \'Esik and Lei\ss\ \cite{esle2002,EsikLeiss05} consider \emph{$\mu$-expressions} defined by the grammar
\begin{align*}
t &::= x \mid t + t \mid t \cdot t \mid 0 \mid 1 \mid \mu x. t
\end{align*}
where $x\in X$. These expressions provide a syntax with which least solutions of polynomial systems can be named. Scope, bound and free occurrences of variables, $\alpha$-conversion, and safe substitution are defined as usual (see e.g.\ \cite{Barendregt84}). We denote by $t\subst xu$ the result of substituting $u$ for all free occurrences of $x$ in $t$, renaming bound variables as necessary to avoid capture. Let $\T X$ denote the set of $\mu$-expressions over indeterminates $X$. 

Let $C$ be a Chomsky algebra and $X$ a set of indeterminates.
An \emph{interpretation} over $C$ is a map
\fhcomment{Not a semiring homomorphism, since domain $\T X$ not a semiring.}%
\dkcomment{OK, how about ``homomorphism with respect to the semiring operators''}%
\fhcomment{OK.}%
$\sigma:\T X\fun C$ that is a homomorphism with respect to the semiring operations and 
such that
\begin{align}
\sigma(\mu x. t) &= \text{the least $a\in C$ such that $\sigma\rebind xa(t) \leq a$},\label{ax:mu}
\end{align}
where $\sigma\rebind xa$ denotes $\sigma$ with $x$ rebound to $a$.
The element $a$ exists and is unique: Informally,
each $\mu$-expression $t$ can be associated with a system of polynomial inequalities such that $\sigma(t)$ is a designated component of its least
solution, which exists by algebraic closure.
\fhcomment{Change to: by Beki\'c's Theorem.}%
\dkcomment{No, it's by algebraic closure. Nothing to do with Beki\'c.}%
\fhcomment{The $p$ here is from $\T X$, which may contain occurrences of $\mu$.  The definition of algebraic closure only requires least solutions for polynomials from $C[X]$.  Isn't it a bit brief to say that this, existence of least solutions, also holds for $\T X$ by algebraic closure (which I'd read as ``follows immediately from the definition of algebraic closure'')?}%

Every set map $\sigma:X\fun C$ extends uniquely to such a homomorphism. An interpretation $\sigma$ \emph{satisfies} the equation $s=t$ if $\sigma(s)=\sigma(t)$ and satisfies the inequality $s\leq t$ if $\sigma(s)\leq\sigma(t)$. All interpretations over Chomsky algebras satisfy the axioms of idempotent semirings, $\alpha$-conversion (renaming of bound variables), and the \emph{Park axioms}
\fhcomment{Add: Reference to where they are introduced.}%
\dkcomment{??}%
\fhcomment{First or canonical piece of literature that contains the 
Park axioms.  Presumably something Park has written and that subsequently was called Park by somebody else in another paper.  Couldn't find those references, and I'm curious...}%
\begin{align}
t\subst x{\mu x.t} &\leq \mu x.t & t\leq x\ &\Imp\ \mu x.t \leq x.\label{ax:Park}
\end{align}
The Park axioms say intuitively that $\mu x.t$ is the least solution of the single inequality $t\leq x$.
\fhcomment{Change to: names the least solution}%
\dkcomment{Ah, the classic dilemma of use vs. mention. Technically you are right but it's a common abuse in algebra and logic. I personally would prefer to stick with ``is'' but willing to go with ``names'' if it makes you happy. But, realize, if we do this, we technically should change the next statement to ``$p\subst x{\mu x.p}$ and $\mu x.p$ name the same element under any interpretation.''}%
\fhcomment{OK.  I was suggesting ``names'' only as linguistic connection to the occurrence of ``name'' in ``to name least solutions'' below.}%
It follows easily that
\begin{align}
t\subst x{\mu x.t} &= \mu x.t.\label{ax:Park1}
\end{align}

Thus Chomsky algebras are essentially the ordered Park $\mu$-semirings of \cite{EsikLeiss05} with the additional restriction that $+$ is idempotent and the order is the natural order $x\leq y \Iff x + y = y$.

\subsection{Beki\'c's Theorem}
\label{sec:Bekic}

It is well known that the ability to name least solutions of single inequalities with $\mu$ gives the ability to name least solutions of all finite systems of inequalities. This is known as Beki\'c's theorem \cite{bekic1984}. The construction is analogous to the definition of $M^\star$ for a matrix $M$ over a Kleene algebra. 
\fhcomment{Add: Formulate Beki\'c's Theorem: Algebraic closure for systems of inequations over $\E X$, not just $\P X$.}%
\dkcomment{I suggest not wasting space with this. It is well known and done in detail in Winskel's book and \'Esik and Leiss, and we give the references below. We also do the 2x2 example, which is more instructive than any general formulation would be.}%
\fhcomment{OK.}%

Beki\'c's theorem can be proved by regarding a system of inequalities as a single inequality on a Cartesian product, partitioning into two systems of smaller dimension, then applying the result for the $2\times 2$ case inductively. The $2\times 2$ system
\begin{align*}
p(x,y) &\leq x & q(x,y) &\leq y
\end{align*}
has least solution $a_0,b_0$, where
\begin{align*}
a(y) &= \mu x.p(x,y) & b_0 &= \mu y.q(a(y),y) & a_0 &= a(b_0),
\end{align*}
as can be shown using the Park axioms \eqref{ax:Park}; see \cite{Winskel93} or \cite{EsikLeiss05} for a comprehensive treatment.

For example, in the context-free languages, the set of strings in $\{a,b\}^\star$ with equally many $a$'s and $b$'s is represented by the term
\begin{align}
\mu S.(1+a\cdot \mu B.(bS+aBB) + b\cdot \mu A.(aS+bAA))\label{eq:cfexample}
\end{align}
obtained from the system \eqref{ex:Greibach1} by this construction.

\subsection{$\mu$-Continuity}
\label{sec:continuity}

Let $nx.t$ be an abbreviation for the $n$-fold composition of $t$ applied to $0$, defined inductively by
\begin{align*}
0x.t &= 0 & (n{+}1)x.t &= t\subst x{nx.t}.
\end{align*}
A Chomsky algebra is called \emph{$\mu$-continuous} if it satisfies the \emph{$\mu$-continuity axiom}:
\begin{align}
a(\mu x.t)b &= \sum_{n\geq 0}\,a(nx.t)b,\label{ax:mucont}
\end{align}
where the summation symbol denotes supremum with respect to the natural order $x \leq y 
\Iff x+y=y$. Note that the supremum of $a$ and $b$ is $a+b$.

\newcommand\canon[1]{L_{#1}}

The family $\CF X$ of context-free languages over an alphabet $X$ forms a $\mu$-continuous Chomsky algebra. The \emph{canonical interpretation} over this algebra is $\canon X:\T X\fun\CF X$, where
\begin{align}
\canon X(x) &= \{x\} &
\canon X(t + u) &= \canon X(t)\cup \canon X(u)\nonumber\\
\canon X(0) &=\emptyset &
\canon X(tu) &= \set{xy}{x\in \canon X(t),\ y\in \canon X(u)}\label{ax:Ldef}\\
\canon X(1) &= \cset{\eps} &
\canon X(\mu x.t) &= \bigcup_{n\geq 0} \canon X(nx.t).\nonumber
\end{align}
Under $\canon X$, every term in $\T X$ represents a context-free language over its free variables (note that $x$ is not free in $nx.t$). In the example \eqref{eq:cfexample} of \S\ref{sec:Bekic}, the free variables are $a,b$ and the bound variables are $S,A,B$, corresponding to the terminal and nonterminal symbols, respectively, of the grammar \eqref{ex:Greibach1} of \S\ref{sec:Chomsky}.

\subsection{Relation to Other Axiomatizations}

In this section we show that the various axiomatizations considered in \cite{esle2002,EsikLeiss05,leiss92b} are valid in all $\mu$-continuous Chomsky algebras.

A \emph{$\mu$-semiring} \cite{EsikLeiss05} is a semiring $(A,+,\cdot,0,1)$ satisfying the \emph{$\mu$-congruence} and \emph{substitution} properties:
\begin{align*}
t=u &\Imp \mu x.t=\mu x.u & \sigma(t\subst yu) = \sigma\rebind y{\sigma(u)}(t).
\end{align*}
    Idempotence is not assumed.
\begin{lemma}\label{lem:chomskyismusemi}
    Every Chomsky algebra is a $\mu$-semiring.
\end{lemma}
\begin{proof}
    The $\mu$-congruence property is immediate from the definition of the $\mu$ operation \eqref{ax:mu}. 
    The substitution property is a general property of systems with variable bindings; see \cite[Lemma 5.1.5]{Barendregt84}. It can be proved by induction. For the case of $\mu x.t$, we assume without loss of generality that $y\neq x$ (otherwise there is nothing to prove) and that $x$ is not free in $u$.
\begin{align*}
\sigma((\mu x.t)\subst yu) &= \sigma(\mu x.(t\subst yu))\\
&= \text{least $a$ such that $\sigma\rebind xa(t\subst yu) \leq a$}\\
&= \text{least $a$ such that $\sigma\rebind xa\rebind y{\sigma(u)}(t) \leq a$}\\
&= \text{least $a$ such that $\sigma\rebind y{\sigma(u)}\rebind xa(t) \leq a$}\\
&= \sigma\rebind y{\sigma(u)}(\mu x.t).
\end{align*}
\end{proof}

We now consider various axioms proposed in \cite{leiss92b}.

\begin{lemma}\label{lem:samelists}
In all $\mu$-continuous Chomsky algebras,
\begin{align*}
\mu x.(1+ax) = \mu x.(1+xa),\quad x\not\in\FV a.
\end{align*}
\end{lemma}
\begin{proof}
By $\mu$-continuity, it suffices to show that $nx.(1+ax)=nx.(1+xa)$ for all
$n$. We show by induction that for all $n$, $nx.(1+ax)=nx.(1+xa)=\sum_{i=0}^n a^i$.
The basis $n=0$ is trivial. For the inductive case,
\begin{align*}
(n{+}1)x.(1+ax) &= 1+a(nx.(1+ax)) 
= 1+a(\textstyle\sum_{i=0}^na^i) 
= \textstyle\sum_{i=0}^{n+1}a^i,
\end{align*}
and this is equal to $(n{+}1)x.(1+xa)$ by a symmetric argument.
\end{proof}

\begin{lemma}\label{lem:listdistr}
The following two equations hold in all $\mu$-continuous Chomsky algebras:
\begin{align*}
a(\mu x.(1+xb)) &= \mu x.(a+xb) & (\mu x.(1+bx))a &= \mu x.(a+bx).
\end{align*}
\end{lemma}
\begin{proof}
We show the first equation only; the second follows from a
symmetric argument.  By $\mu$-continuity, we need only show that the
equation holds for any $n$.  The basis $n=0$ is trivial.  For the inductive case,
\begin{align*}
a((n{+}1)x.(1+xb)) &= a + a(nx.(1+xb))b \\
&= a + (nx.(a+xb))b \\
&= (n{+}1)x.(a+xb),
\end{align*}
where the induction hypothesis has been used in the second step.
\end{proof}
These properties also show that $\mu$-continuous Chomsky algebras are algebraically complete semirings in the sense of \cite{esle2002,EsikLeiss05}.
\begin{lemma}
The \emph{Greibach inequalities}
\begin{align*}
\mu x.s(\mu y.(1+ry)) &\leq \mu x.(s+xr) &
\mu x.(\mu y.(1+yr))s &\leq \mu x.(s+rx)
\end{align*}
of \KAG~{\upshape\cite{leiss92b}} hold in all $\mu$-continuous Chomsky algebras.
\end{lemma}
\begin{proof}
For the left-hand inequality, let $u=\mu x.(s+xr)$. By the Park axioms, it suffices to show that $s(\mu y.(1+ry))\subst xu \leq u$. But
\begin{align*}
s(\mu y.(1+ry))\subst xu &= s\subst xu(\mu y.(1+r\subst xu y)) \\
&= s\subst xu (\mu y.(1+yr\subst xu)) \\
&= \mu y.(s\subst xu + yr\subst xu) \\
&= \mu x.(s+xr),
\end{align*}
where Lemmas \ref{lem:samelists} and \ref{lem:listdistr} have been used.

The right-hand ineuuality can be proved by a symmetric argument.
%
\end{proof}

Various other axioms of \cite{esle2002,EsikLeiss05,leiss92b} follow from the Park axioms.

The $\mu$-continuity condition \eqref{ax:mucont} implies the Park axioms \eqref{ax:Park}, but we must defer the proof of this fact until \S\ref{sec:main}. For now we just observe a related property of the canonical interpretation $\canon X$.

\begin{lemma}
\label{lem:freecont}
For any $s,t\in\T X$ and $y\in X$,
\fhcomment{Change: $p$ to $s$; $p$ is reminiscent of ``polynomial''.}%
\dkcomment{OK}%
\begin{align*}
\canon X(s\subst y{\mu y.t}) = \bigcup_{n\geq 0}\canon X(s\subst y{ny.t}).
\end{align*}
\end{lemma}
\begin{proof}
We proceed by induction on the structure of $s$.
The cases for $+$ and $\cdot$ are quite easy, using the facts that for chains of sets of strings $A_0\subs A_1\subs A_2\subs\cdots$ and $B_0\subs B_1\subs B_2\subs\cdots$,
\begin{align*}
\bigcup_m A_m\cup\bigcup_n B_n &= \bigcup_n A_n\cup B_n &
\bigcup_m A_m\cdot\bigcup_n B_n &= \bigcup_n A_nB_n.
\end{align*}
The base cases are also straightforward. For $\mu x.s$, assume without loss of generality that $y\neq x$ and $x$ is not free in $t$.
\begin{align*}
\canon X((\mu x.s)\subst y{\mu y.t}) &= \bigcup_m\canon X((mx.s)\subst y{\mu y.t})\\
&= \bigcup_m\bigcup_n\canon X((mx.s)\subst y{ny.t})\\
&= \bigcup_n\bigcup_m\canon X((mx.s)\subst y{ny.t})\\
&= \bigcup_n\canon X((\mu x.s)\subst y{ny.t}).
\end{align*}

\end{proof}

%
%
%

\section{Main Results}
\label{sec:main}

Our main result depends on an analog of a result of \cite{kozen81} (see \cite{K91a}). It asserts that the supremum of a context-free language over a $\mu$-continuous Chomsky algebra $K$ exists, interpreting strings over $K$ as products in $K$. Moreover, multiplication is continuous with respect to suprema of context-free languages.

\begin{lemma}
\label{lem:cflsupr}
Let $\sigma:\T X\fun K$ be any interpretation over a $\mu$-continuous Chomsky algebra $K$.
Let $\tau:\T X\fun\CF X$ be any interpretation over the context-free languages $\CF X$ such that for all $x\in X$ and $s,u\in\T X$,
\begin{align*}
\sigma(sxu) &= \sum_{y\,\in\,\tau(x)} \sigma(syu).
\end{align*}
Then for any $s,t,u\in\T X$,
\begin{align*}
\sigma(stu) &= \sum_{y\,\in\,\tau(t)} \sigma(syu).
\end{align*}
In particular,
\begin{align}
\sigma(stu) &= \sum_{y\in\canon X(t)} \sigma(syu),\label{eq:cflsuprL}
\end{align}
where $\canon X$ is the canonical interpretation defined in \S\ref{sec:continuity}.
\end{lemma}
\begin{remark}
Note carefully that the lemma does not assume \emph{a priori} knowledge of the existence of the suprema. The equations should be interpreted as asserting that the supremum on the right-hand side exists and is equal to the expression on the left-hand side.
\end{remark}
\begin{proof}
The proof is by induction on the structure of $t$, that is by
induction on the subexpression relation $t + u \succ t, t + u \succ u,
t \cdot u \succ t, t \cdot u \succ u, \mu x. t \succ n x. t$, which is 
well-founded \cite{kozen83}.

All cases are similar to the proof in \cite[Lemma 7.1]{K91a} for star-continuous Kleene algebra, with the exception of the case $t=\mu x.p$.

For variables $t=x\in X$, the desired property holds by assumption. For the constants $t=0$ and $t=1$,
\begin{align*}
\sigma(s0u) &= 0 = \sum\,\emptyset = \sum_{y\,\in\,\emptyset} \sigma(syu) = \sum_{y\,\in\,\tau(0)} \sigma(syu)\\
\sigma(s1u) &= \sigma(su) = \sum_{y\,\in\,\cset\eps} \sigma(syu) = \sum_{y\,\in\,\tau(1)} \sigma(syu).
\end{align*}

For sums $t=p+q$,
\begin{align}
\sigma(s(p+q)u) &= \sigma(spu) + \sigma(squ)\nonumber\\
&= \sum_{x\,\in\,\tau(p)} \sigma(sxu) + \sum_{y\,\in\,\tau(q)} \sigma(syu)\label{lem:p+q1}\\
&= \sum_{z\,\in\,\tau(p)\cup\tau(q)} \sigma(szu)\label{lem:p+q2}\\
&= \sum_{z\,\in\,\tau(p+q)} \sigma(szu)\label{lem:p+q3}.
\end{align}
Equation \eqref{lem:p+q1} is by two applications of the induction hypothesis.
Equation \eqref{lem:p+q2} is by the properties of supremum. 
Equation \eqref{lem:p+q3} is by the definition of sum in $\CF X$. 

For products $t=pq$,
\begin{align}
\sigma(spqu) &= \sum_{x\,\in\,\tau(p)}\,\sum_{y\,\in\,\tau(q)} \sigma(sxyu)\label{lem:pq1}\\
&= \sum_{z\,\in\,\tau(p)\cdot\tau(q)} \sigma(szu)\label{lem:pq2}\\
&= \sum_{z\,\in\,\tau(pq)} \sigma(szu)\label{lem:pq3}.
\end{align}
Equation \eqref{lem:pq1} is by two applications of the induction hypothesis.
Equations \eqref{lem:pq2} and \eqref{lem:pq3} are by the definition of product in $\CF X$. 

Finally, for $t=\mu x.p$,
\begin{align}
\sigma(s(\mu x.p)u) &= \sum_n\,\sigma(s(nx.p)u)\label{lem:cflsupr1}\\
&= \sum_n \sum_{y\,\in\,\tau(nx.p)} \sigma(syu)\label{lem:cflsupr2}\\
&= \sum_{y\,\in\,\bigcup_n\!\tau(nx.p)} \sigma(syu)\label{lem:cflsupr3}\\
&= \sum_{y\,\in\,\tau(\mu x.p)} \sigma(syu).\label{lem:cflsupr4}
\end{align}
Equation \eqref{lem:cflsupr1} is just the $\mu$-continuity property \eqref{ax:mucont}.
Equation \eqref{lem:cflsupr2} is by the induction hypothesis, observing that $\mu x. p \succ nx.p$. 
Equation \eqref{lem:cflsupr3} is a basic property of suprema. 
Finally, equation \eqref{lem:cflsupr4} is by the definition of $\tau(\mu x.p)$ in $\CF X$.

The result \eqref{eq:cflsuprL} for the special case of $\tau=\canon X$ is immediate, observing that $\canon X$ satisfies the assumption of the lemma: for $x\in X$,
\begin{align*}
\sigma(sxu) &= \sum_{y\in\cset x}\,\sigma(syu) = \sum_{y\in\canon X(x)}\,\sigma(syu).
\end{align*}
\end{proof}

At this point we can show that the $\mu$-continuity condition implies the Park axioms.
\begin{theorem}
The $\mu$-continuity condition \eqref{ax:mucont} implies the Park axioms \eqref{ax:Park}.
\end{theorem}
\begin{proof}
We first show $p\leq x\Imp\mu x.p\leq x$ in any idempotent semiring satisfying the $\mu$-continuity condition. Let $\sigma$ be a valuation such that $\sigma(\mu x.p) = \sum_n \sigma(nx.p)$. Suppose that $\sigma(p)\leq \sigma(x)$. We show by induction that for all $n\geq 0$, $\sigma(nx.p)\leq \sigma(x)$. This is certainly true for $0x.p = 0$. Now suppose it is true for $nx.p$. Using monotonicity,
\begin{align*}
\sigma((n{+}1)x.p) = \sigma(p\subst x{nx.p}) \leq \sigma(p\subst xx) = \sigma(p) \leq \sigma(x).
\end{align*}
By $\mu$-continuity, $\sigma(\mu x.p) = \sum_n \sigma(nx.p) \leq \sigma(x)$.

Now we show that $p\subst x{\mu x.p} \leq \mu x.p$. This requires the stronger property that a $\mu$-expression is chain-continuous with respect to suprema of context-free languages as a function of its free variables. Using Lemmas \ref{lem:freecont} and \ref{lem:cflsupr},
\begin{align*}
\sigma(p\subst x{\mu x.p}) &= \sum\,\set{\sigma(y)}{y\in\canon X(p\subst x{\mu x.p})}\\
&= \sum\,\set{\sigma(y)}{y\in\bigcup_n\canon X(p\subst x{nx.p})}\\
&= \sum_n\,\sum\,\set{\sigma(y)}{y\in\canon X(p\subst x{nx.p})}\\
&= \sum_n\,\sigma(p\subst x{nx.p})\\
&= \sum_n\,\sigma((n{+}1)x.p)\\
&= \sigma(\mu x.p).
\end{align*}
\end{proof}

The following is our main theorem.
\begin{theorem}
Let $X$ be an arbitrary set and let $s,t\in\T X$. 
\fhcomment{Question: Which $X$? Arbitrary set?}%
\dkcomment{Yes.}%
The following are equivalent:
\begin{enumerate}
\renewcommand\labelenumi{{\upshape(\roman{enumi})}}
\item
The equation $s=t$ holds in all $\mu$-continuous Chomsky algebras; that is, $s=t$ is a logical consequence of the axioms of idempotent semirings and the $\mu$-continuity condition
\begin{align}
a(\mu x.t)b &= \sum_{n\geq 0}\,a(nx.t)b,\label{ax:mucontA}
\end{align}
or equivalently, the universal formulas
\begin{gather}
a(nx.t)b \leq a(\mu x.t)b,\quad n\geq 0\label{ax:mucont1A}\\
\left(\bigwedge_{n\geq 0} (a(nx.t)b \leq w)\right)\ \Imp\ a(\mu x.t)b \leq w.\label{ax:mucont2A}
\end{gather}
\item
The equation $s=t$ holds in the semiring of context-free languages $\CF Y$ over any set $Y$.
\item
$\canon X(s)=\canon X(t)$, where $\canon X:\T X\fun\CF X$ is the standard interpretation mapping a $\mu$-expression to a context-free language of strings over its free variables.
\end{enumerate}
Thus the axioms of idempotent semirings and $\mu$-continuity are sound and complete for the equational theory of the context-free languages.
\end{theorem}
\begin{proof}
The implication (i) $\Imp$ (ii) holds since $\CF Y$ is a $\mu$-continuous Chomsky algebra, and (iii) is a special case of (ii). Finally, if (iii) holds, then by two applications of Lemma \ref{lem:cflsupr}, for any interpretation $\sigma:\T X\fun K$ over a $\mu$-continuous Chomsky algebra $K$,
\begin{align*}
\sigma(s) &= \sum_{x\in\canon K(s)} \sigma(x) = \sum_{x\in\canon K(t)} \sigma(x) = \sigma(t),
\end{align*}
which proves (i).
\end{proof}

\begin{theorem}
The context-free languages over the alphabet $X$ form the free $\mu$-continuous Chomsky algebra on generators $X$.
\end{theorem}
\begin{proof}
Let $K$ be a $\mu$-continuous Chomsky algebra.
Any map $\sigma:X\fun K$ extends uniquely to an interpretation $\sigma:\T X\fun K$. By Lemma \ref{lem:cflsupr}, this decomposes as
\begin{align*}
\sigma &= {\textstyle\sum}\circ\CF\sigma\circ\canon X,
\end{align*}
where $\canon X:\T X\fun\CF X$ is the canonical interpretation in the context-free languages over $X$, $\CF\sigma:\CF X\fun\CF K$ is the map $\CF\sigma(A) = \set{\sigma(x)}{x\in A}$, and ${\sum}:\CF K\fun K$ takes the supremum of a context-free language over $K$, which is guaranteed to exist by Lemma \ref{lem:cflsupr}. The unique morphism $\CF X\fun K$ corresponding to $\sigma$ is ${\sum}\circ\CF\sigma$. Thus $\CFname$ is left adjoint to the forgetful functor from $\mu$-continuous Chomsky algebras to \Set. The maps $x\mapsto\cset x:X\fun\CF X$ and ${\sum}:\CF K\fun K$ are the unit and counit, respectively, of the adjunction.
\end{proof}

\section{Conclusion}

We have given a natural complete infinitary axiomatization of the equational theory of the context-free languages. Lei\ss\ \cite{leiss92b} states as an open problem:
\begin{quote}
Are there natural equations between $\mu$-regular expressions that are valid in all
continuous models of \KAF, but go beyond \KAG?
\end{quote}
We have identified such a system in this paper, thereby answering Lei\ss's question. He does not state axiomatization as an open problem, but observes that the set of pairs of equivalent context-free grammars is not recursively enumerable, then goes on to state:
\begin{quote}
Since there is an effective translation between context-free grammars and $\mu$--regular expressions\,\ldots, the equational theory of context-free languages in terms of $\mu$-regular expressions is not axiomatizable at all.
\end{quote}
Nevertheless, we have given an axiomatization.
How do we reconcile these two views?
Lei\ss\ is apparently using ``axiomatization'' in the sense of ``recursive axiomatization.'' But observe that the axiom
\fhcomment{Change: Remove ``such'' since Lei\ss\ obviously uses
``axiomtization'' in the sense of ``recursive axiomatization''.  Add a statement to this effect (i.e. Lei\ss\ implicitly meaning ``recursive''.}%
\dkcomment{OK, how about: Nevertheless, we have given an axiomatization.
How do we reconcile these two views?
Lei\ss\ is apparently using ``axiomatization'' in the sense of ``recursive axiomatization''. But observe that the axiom...}%
\fhcomment{I like it. It makes clear that we are not accusing Lei\ss\ of 
making a technically wrong statement, but of not making the distinction between recursive and unrestricted axiomatization.  No need to appear to want to insult him, given he is bound to referee our paper...}%
\dkcomment{Unfortunately, this kind of blunts the impact of the statement. If Leiss had meant ``recursive'', he should have said it.}%
\fhcomment{It's pretty clear he meant it, and, yes, he should have said ``recursive''.  He is not the first one to mentally conflate ``recursive'' into ``axiomatization''.   Committing the mortal sin of 
formulating an inference rule with an undecidable side condition and calling the result an axiomatization can get one into trouble :-).}%
%
\eqref{ax:mucont2A} is an infinitary Horn formula. To use it as a rule of inference, one would need to establish infinitely many premises of the form $x(ny.p)z \leq w$. But this in itself is a $\Pi_1^0$-complete problem. One can show that it is $\Pi_1^0$-complete to determine whether a given context-free grammar $G$ over a two-letter alphabet generates all strings. 
\fhcomment{Add textbook reference? E.g.~to your book?}%
By coding $G$ as a $\mu$-expression $w$, the problem becomes $\mu x.(1 + ax + bx) \leq w$, which by \eqref{ax:mucontA} is equivalent to showing that $nx.(1 + ax + bx)\leq w$ for all $n$.
\fhcomment{This is cool!}%

\section*{Acknowledgments}

We thank Zolt{\'an} {\'E}sik, Hans Lei\ss, and the anonymous referees for helpful comments.  The DIKU-affiliated authors express their thanks to the Department of Computer Science at Cornell University for hosting them in the Spring 2013 and to the Danish Council for Independent Research for financial support for this work under Project 11-106278, ``Kleene Meets Church (KMC): Regular Expressions and Types''.

\bibliographystyle{eptcs}
\bibliography{CFL}

\end{document}